\documentclass[a4paper,10pt]{article}
\usepackage{color}
\usepackage{graphicx}
\usepackage[affil-it]{authblk}
\usepackage{amssymb,amsthm,amsfonts,amstext}
\usepackage{url}
\def\be{\begin{equation}}
\def\ee{\end{equation}}
\def\bea{\begin{eqnarray}}
\def\eea{\end{eqnarray}}
\def\bma{\begin{mathletters}}
\def\ema{\end{mathletters}}

\def\0{\overline{0}}
\def\Tr{\mbox{Tr}}

\def\q0{\underline{0}}

\def\H{{\cal H}}
\def\S{{\cal S}}

\def\id{{\mathbb I}}

\def\L{{\cal L}}
\def\Q{{\cal Q}}
\def\GW{{\cal GW}}
\def\TW{{\cal TW}}

\def\R{\mathbb{R}}
\def\N{\mathbb{N}}

\def\tr{\mbox{tr}}

\def\TOBL{{\cal TOBL}}
\def\NS{{\cal NS}}
\def\one{\leavevmode\hbox{\small1\normalsize\kern-.33em1}}

\def\bra#1{\langle#1|} \def\ket#1{|#1\rangle}

\newtheorem{theo}{Theorem}
\newtheorem{defin}{Definition}

\newtheorem{lemma}{Lemma}

\newtheorem{conj}{Conjecture}
\newtheorem{quest}{Question}
\newtheorem{prob}{Problem}

\def\id{{\mathbb I}}
\def\tr{\mbox{tr}}

\begin{document}

\title{Closed sets of correlations: answers from the zoo}
\author[1]{Ben Lang}
\author[2]{Tam\'as V\'ertesi}
\author[3]{Miguel Navascu\'es}
\affil[1]{School of Physics, University of Bristol, Tyndall Avenue, Bristol, BS8 1TL, United Kingdom}
\affil[2]{Institute for Nuclear Research, Hungarian Academy of Sciences, H-4001 Debrecen, P.O. Box 51, Hungary}
\affil[3]{Universitat Aut\`onoma de Barcelona, 08193 Bellaterra (Barcelona), Spain}

\begin{abstract}
We investigate the conditions under which a set of multipartite nonlocal correlations can describe the distributions achievable by distant parties conducting experiments in a consistent universe. Several questions are posed, such as: are all such sets ``nested'', i.e., contained into one another? Are they discrete or do they form a continuum? How many of them are supraquantum? Are there non-trivial polytopes among them? We answer some of these questions or relate them with established conjectures in complexity theory by introducing a ``zoo'' of physically consistent sets which can be characterized efficiently via either linear or semidefinite programming. As a bonus, we use the zoo to derive, for the first time, concrete impossibility results in nonlocality distillation.

\end{abstract}

\maketitle

\section{Introduction}

The lack of a physical intuition behind the fundamental axioms under which Quantum Theory rests (states are rays of a Hilbert space, etc.) has led many to wonder whether these axioms are actually necessary to construct a physical theory, or, on the contrary, different, weirder, theories exist, at least at the logical level. This reflection has inspired an ambitious program to reconstruct Quantum Mechanics from physical principles, see \cite{hardy,axioms, lluis_mas, chiribella, hardy2} for some impressive achievements in this topic. An alternative approach has been to isolate those features which make Quantum Mechanics special, such as nonlocality, and investigate if similar or even stranger phenomenons can be found in other physical theories. Popescu \& Rohrlich inaugurated the latter line of research by proposing the no-signalling principle as an attempt to bound the set of feasible correlations that two or more distant parties can establish. As they showed, the no-signalling condition is not strong enough to single out the quantum set \cite{PR_box}. This later caused a proliferation of device-independent physical principles (such as  Non-trivial Communication Complexity \cite{brassard}, No Advantage for Nonlocal Computation \cite{linden}, Information Causality \cite{info_caus}, Macroscopic Locality \cite{mac_loc} and Local Orthogonality \cite{loc_orth}) which have since constrained the set of physically admissible correlations further and further.

In \cite{closed} it was noted that, no matter which principles a physical theory satisfies, its associated set of correlations must be \emph{closed under wirings}, meaning that any combination or wiring of different valid distributions can only produce boxes inside the considered set. Contrary to all expectations, the authors of \cite{closed} proved that this seemingly innocuous concept is highly non-trivial, and thus may play an important role in the axiomatization of quantum mechanics, or the exploration of alternative theories.

Despite the fact that five years have passed since the concept was coined, we still ignore many facts about consistent sets of correlations. Even though investigations into nonlocality distillation suggest that there exists a continuum of such sets \cite{short,forster}, very few concrete examples of provenly closed sets are known. The original paper \cite{closed} just identifies four of them, although further results demonstrate (in a non-constructive way) the existence of infinitely many \cite{dukaric}.

We feel that the lack of progress in this subject is due in part to the scarcity of examples of computable closed sets. In this paper, we intend to fill this gap by introducing several families of closed sets which can be efficiently characterized using standard tools of convex optimization. Using these sets, we manage to provide an answer to many important questions regarding this fascinating topic. Along the way, we also prove the existence of bipartite physical principles which are \emph{unstable under composition}, in the sense that two boxes living in different theories compatible with such principles can be wired together to produce a box violating them.

The structure of this paper is as follows: first, in Section \ref{setup} we will describe the non-locality framework that we will be invoking through this text. Then, in Section \ref{questions}, we will review known results on consistent sets of correlations, and we will list some relevant questions on the topic. In Section \ref{zoo} we will introduce our ``zoo'' of consistent sets admitting an efficient characterization. These sets will allow us to answer many of the questions previously posed in Section \ref{answers}. Finally, in Section \ref{conclusion}, we will present our conclusions.

\section{The nonlocality framework and consistent sets of correlations}
\label{setup}

\subsection{The setup}
Let Alice and Bob be two parties conducting experiments in distant laboratories. We will assume that Alice and Bob ignore the inner workings of their measurement devices: for Alice (Bob), an experiment is a process or black box to which she (he) feeds an input $x$ ($y$) from the alphabet ${\cal X}$ (${\cal Y}$), and from which she (he) receives an output $a$ ($b$) from the alphabet ${\cal A}$ (${\cal B}$). Along this article, we will consider Bell scenarios where ${\cal X},{\cal Y},{\cal A},{\cal B}$ have finite cardinality. If Alice and Bob compare their outputs in independent runs of the experiment, then they can estimate the probabilities $P(a,b|x,y)$ that define their pair of correlated boxes.

\noindent An important set of such boxes is the set of all boxes compatible with classical physics.

\begin{defin}\textbf{The local set $\L$}\\
We say that $P(a,b|x,y)$ is \emph{local} or \emph{classical} when it can be expressed as

\be
P(a,b|x,y)=\sum_{\lambda}P(\lambda)P(a|x,\lambda)Q(b|y,\lambda),
\ee

\noindent with $P(\lambda)\geq 0$, $\sum_{\lambda}P(\lambda)=1$.

\end{defin}

\noindent The set $\L$ of all local distributions is a polytope (a convex set with finitely many vertices) whose extreme points are given by the deterministic boxes $P(a,b|x,y)=\delta_{a,f(x)}\delta_{b,g(y)}$. As such, it can be characterized via linear programming (LP) \cite{linear_prog}.

A more complicated, but perhaps more relevant set is $\Q$, the set of all boxes realizable with quantum mechanical systems.

\begin{defin}{\textbf{The quantum set $\Q$}}\\
A distribution $P(a,b|x,y)$ is quantum iff

\be
P(a,b|x,y)=\tr(\rho_{AB}E^x_a\otimes F^y_b),
\ee

\noindent where $\rho_{AB}\in B(\H_A\otimes\H_B)$ is a normalized quantum state acting over the tensor product of the Hilbert spaces $\H_A,\H_B$, and $\{E^x_a\}\subset B(\H_A)$ ($\{F^y_b\}\subset B(\H_B)$) are projector operators satisfying the completeness relations $\sum_a E^x_a=\id_A$, $\sum_b F^y_b=\id_B$.
\end{defin}

\noindent The set $\Q$ is not a polytope: it has both straight and curved surfaces, and its characterization is not known to be a decidable problem \cite{decidable}, although there are algorithms to bound it from the outside \cite{quantum1, quantum2} via semidefinite programming \cite{sdp} and the inside \cite{I3322_PV} via see-saw methods.

One can go beyond our present understanding of the universe, and consider also boxes which do not admit a quantum representation at all. In this respect, an interesting set of boxes is the one resulting from demanding the \emph{no-signalling conditions} \cite{PR_box} to hold:

\begin{defin}{\textbf{The no-signalling set $\NS$}}\\

$P(a,b|x,y)$ is no-signalling if it satisfies the no-signalling conditions

\be
\sum_{a}P(a,b|x,y)=P(b|y), \sum_{b}P(a,b|x,y)=P(a|x).
\label{no_sig}
\ee

\end{defin}

\noindent Eq. (\ref{no_sig}) implies that Alice (Bob) cannot modify Bob's (Alice's) statistics by virtue of her (his) input choice, and hence the said pair of boxes does not allow both parties to violate causality. Like $\L$, the set $\NS$ is also a polytope, and thus admits an LP characterization.

By far, the most studied setting in nonlocality is the simplest one where the sets $\L$, $\Q$ and $\NS$ differ: the 2222 Bell scenario, with $|{\cal X}|=|{\cal Y}|=|{\cal A}|=|{\cal B}|=2$, where $x,y,a,b$ are assumed to take values in $\{0,1\}$. In this setting, the local set $\L$ is characterized by the no-signalling conditions and the Clauser-Horne-Shimony-Holt (CHSH) \cite{chsh} inequality (and its permutations):

\be
|\langle X_0Y_0\rangle+\langle X_0Y_1\rangle+\langle X_1Y_0\rangle-\langle X_1Y_1\rangle|\leq 2,
\label{chsh_ineq}
\ee

\noindent where $\langle X_iY_j\rangle=P(a=b|i,j)-P(a\not=b|i,j)$.

Conversely, the extreme points of the no-signalling set $\NS$ are the deterministic points plus all equivalent forms of the Popescu-Rohrlich (PR) box \cite{PR_box}:

\be
\mbox{PR}(a,b|x,y)=\frac{1}{2}\delta_{a\oplus b,x\cdot y}.
\ee

It can be shown that any box in the 2222 scenario can be transformed via local operations and shared randomness into an \emph{isotropic box} of the form

\be
\mbox{PR}^\lambda(a,b|x,y)=\lambda\frac{1}{2}\delta_{a\oplus b,x\cdot y}+(1-\lambda)\frac{1}{4},
\ee

\noindent while keeping the same CHSH value, see \cite{lluis} for a proof. In the segment that goes from $\lambda=0$ to $\lambda=1$, the intervals $\lambda\in [0, \frac{1}{2}]$ and $\lambda\in [0, \frac{1}{\sqrt{2}}]$ correspond, respectively, to the $\L$ and $\Q$ regions \cite{tsi_bound}.

\subsection{Wirings}
Suppose that Alice and Bob, rather than preparing a pair of boxes, produce two independent realizations, see Figure \ref{wired}. Then, they can generate a new effective box pair by combining inputs and outputs of the two boxes and postprocessing the two ``internal'' outcomes, as shown in Figure \ref{wired}. This mechanism of building new boxes by processing outputs and inputs of different local boxes in a sequential way is known as \emph{wiring} \cite{closed}. Note that the identity of the first box to measure may depend in general of the effective input assigned to each party. Also, for three boxes or more, the choice of the next box to measure may depend on all the previous output history. The set of all deterministic wirings thus grows at least as $n!$ with the number $n$ of local boxes. Even worse, a brute-force search over the set of all possible wirings already becomes impractical for $n\approx 5$ boxes.

\begin{figure}
  \centering
  \includegraphics[width=8.5 cm]{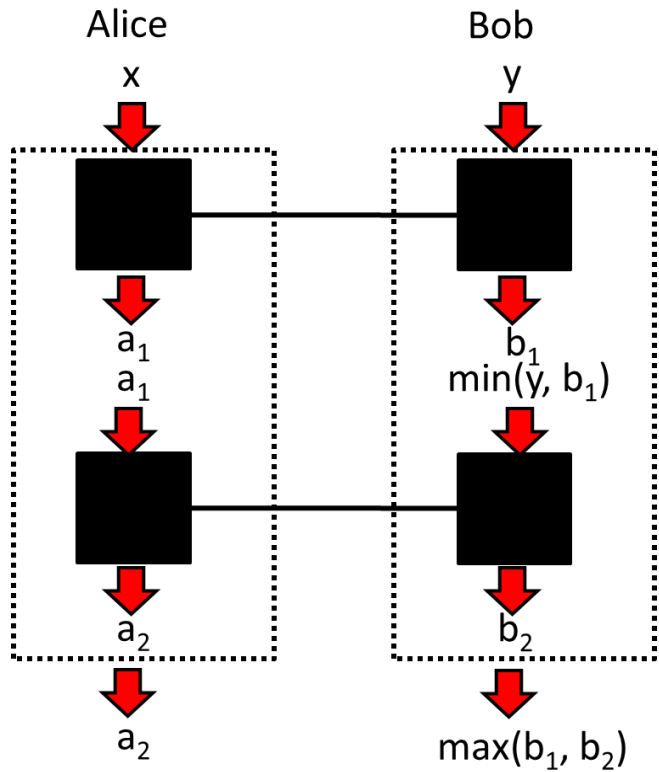}
  \caption{\textbf{Wirings.} By classical circuitry, Alice and Bob can turn two independent pairs of (black) boxes into a new effective pair (dashed line).}
  \label{wired}
\end{figure}

Let us now list the features that a set of correlations ${\cal S}$ must satisfy if ${\cal S}$ is supposed to represent the set of box pairs realizable in a given physical theory. For any pair of boxes $P_1,P_2\in {\cal S}$, Alice and Bob can always prepare the box $\lambda P_1+(1-\lambda)P_2$ if they share some classical randomness: ${\cal S}$ must hence be a convex set. What is more, as noted in \cite{closed}, any possible wiring of a number of box pairs in ${\cal S}$ must be also contained in ${\cal S}$. Finally, Alice and Bob can always interchange their boxes, and so ${\cal S}$ must be symmetric under permutations of the two parties. This reasoning leads to a basic definition on which this work will revolve.

\begin{defin}
Let ${\cal S}$ be a set of boxes. ${\cal S}$ is \emph{physically closed} if

\begin{enumerate}
\item {\cal S} is convex (convexity).
\item {\cal S} is symmetric under the exchange of Alice and Bob (symmetry).
\item For any collection of boxes $\{P_i(a,b|x,y)\}_{i=1}^n\subset {\cal S}$ and any pair of wirings ${\cal W}_A,{\cal W}_B$, the box ${\cal W}_A\otimes {\cal W}_B(\otimes_{i=1}^nP_i)$ belongs to ${\cal S}$ (closure under wirings).

\end{enumerate}

\end{defin}

For simplicity, along the rest of this article, we will refer to physically closed sets simply as \emph{closed sets}. There will be no ambiguity between physical closure and topological closure, since all sets of correlations considered in the article are also topologically closed\footnote{This actually may not be the case for $\Q$. In the following, we will therefore identify this set with its topological closure, and likewise with $\Q_{+}$, defined later.}.

Note that the above definition implies that $\L\subset {\cal S}$. This follows from the convexity of ${\cal S}$ and the observation that, via simple wirings, we can make any pair of boxes deterministic. Notice also that, for any two closed sets $\S_1,\S_2$, one can generate a new closed set by taking their intersection $\S_1\cap \S_2$. Similarly, one can define $\S_1+\S_2$ as the smallest closed set containing $\S_1\cup\S_2$. It is clear that the operator ``$+$'' used here is commutative and associative. The complexity of the operations ``$\cap$'' and ``$+$'', though, is very different: while membership of $\S_1\cap \S_2$ can be easily decided if a characterization of $\S_1$, $\S_2$ is available, determining the limits of $\S_1+\S_2$ can well be an undecidable problem, see the next section.

$\L$, $\Q$ and $\NS$ are distinguished instances of closed sets. Another example is the topological closure of the set of all boxes which can be generated via convex combinations, permutation of the parties and wirings of arbitrarily many copies of a finite set of boxes ${\cal F}$. We will denote such a set as $\S_{{\cal F}}$, and the boxes in ${\cal F}$ will be called \emph{the generators of } $\S_{{\cal F}}$. From all the above, it is clear that

\be
\S_{{\cal F}}+\S_{{\cal G}}=\S_{{\cal F}\cup {\cal G}}.
\ee

\section{Some results and some questions}
\label{questions}

In \cite{closed}, it is shown that fairly natural polytopes in the 2222 scenario\footnote{Namely, the convex hull of the deterministic points plus all re-labelings of the box $PR^\lambda(a,b|x,y)$, for $\frac{1}{2}<\lambda<1$, and the polytope defined by the no-signalling conditions (\ref{no_sig}) plus some relaxed version of the CHSH inequality (\ref{chsh_ineq}) where the numerical coefficient 2 is replaced by a greater amount.} fail to be closed. This evidences that physical closure is an extremely non-trivial property. As a consequence, known closed sets are scarce. In view of the difficulty of defining closed polytopes, the authors of \cite{closed} ask:

\begin{quest}
\label{polytope}
Are there non-trivial closed polytopes?
\end{quest}

Most results on closed sets stem from research on nonlocality distillation \cite{non_loc_dist,short, forster, dukaric}. In the language of closed sets, nonlocality distillation is equivalent to the following membership problem:

\begin{prob}{\textbf{Main Distillability Problem}}\\
Let $P(a,b|x,y),Q(a,b|x,y)\in \NS$. Determine if $P(a,b|x,,y)\in {\cal S}_{\{Q(a,b|x,y)\}}$.
\end{prob}

\noindent The absence of an algorithm to solve the above problem after years of research raises the next question:

\begin{quest}
\label{decidable}
Is the Main Distillability Problem decidable?
\end{quest}

The majority of the research in nonlocality distillation has been focused on the distillation of isotropic PR-boxes. That is, given a number of identical copies of $\mbox{PR}^\lambda(a,b|x,y)$, with $\lambda>\frac{1}{2}$ and shared randomness, the goal is to wire them up so as to build a new isotropic box $\mbox{PR}^{\lambda'}(a,b|x,y)$, with $\lambda'>\lambda$. So far, all attempts to distill isotropic boxes have failed, and this has led to the following widely held conjecture:

\begin{conj}
\label{distill_conj}
For any $1>\mu>\lambda>\frac{1}{2}$, $\mbox{PR}^\mu(a,b|x,y)\not\in \S_{\{\mbox{PR}^\lambda(a,b|x,y)\}}$.
\end{conj}

There are partial results in support of this conjecture. Short proved that no wiring of two isotropic boxes can increase its CHSH value \cite{short}. Using ideas from dynamical programming, this result was later extended by Forster to $n=9$ boxes \cite{forster}. In \cite{dukaric}, Dukaric \& Wolf study CHSH distillation in the quantum nonlocal range $\frac{1}{2}<\lambda<\frac{1}{\sqrt{2}}$. They provide an explicit bound on the amount of possible CHSH distillation, and, as as side result, prove that there are infinitely many values of $\lambda$ in the range $(\frac{1}{2},\frac{1}{\sqrt{2}})$ where distillation is impossible.

Unfortunately, Dukaric and Wolf's proof does not identify any element of the set $\Lambda$ of values of $\lambda$ for which distillation is impossible. Its existence and size nevertheless implies that there are infinitely many different subquantum sets which are closed under wirings. Indeed, define $\Q_\lambda\equiv {\cal S}_{\{\mbox{PR}^\lambda(a,b|x,y)\}}$. Then, for $\lambda,\lambda'\in (\frac{1}{2},\frac{1}{\sqrt{2}})$ with $\lambda>\lambda'$, the sets $\Q_{\lambda}$ and $\Q_{\lambda'}$ satisfy $\L \subsetneq \Q_{\lambda'}\subsetneq \Q_{\lambda'}\subsetneq \Q$.

From the proof it is also not clear whether $\Lambda$ is discrete (countable), or continuous. This is an important matter: note that, if the set of all closed sets were discrete or ``quantized'', then all physical theories could be \emph{numbered} according to which set of correlations they occupy! That would give a lot of structure/hope to the task of classifying all closed sets. We therefore find imperative to answer the question below.

\begin{quest}
\label{continuum}
Is there a continuum of sets which are closed under wirings?
\end{quest}

\noindent Note that, if Conjecture \ref{distill_conj} is true, the above question must be answered in the affirmative.

Dukaric and Wolf's result postulates the existence of non-trivial sets closed under wirings, but does not offer concrete limits of nonlocality distillation. Exploring the present literature in nonlocality, though, one can find some other instances of closed subquantum sets. Consider, for example, the set of all bipartite correlations achievable by two parties conducting quantum measurements over any number of copies of the maximally entangled state $\ket{\psi^+}=\frac{1}{\sqrt{2}}(\ket{00}+\ket{11})$:

\begin{defin}{\textbf{$\Q_{+}$}}\\
$P(a,b|x,y)$ belongs to $\Q_{+}$ iff there exist positive measurement operators $E^x_a,F^y_b\geq 0$, with $\sum_aE^x_a=\sum_bF^y_b=\id$, such that:
\be
P(a,b|x,y)=\bra{\psi^+}^{\otimes N}E^x_a\otimes F^y_b\ket{\psi^+}^{\otimes N}.
\ee

\end{defin}

Clearly, any wiring of a collection of boxes in $\Q_{+}$ can be expressed as a generalized measurement over a number of maximally entangled states, i.e., $\Q_{+}$ is closed under wirings. Also, $Q_{+}\subset Q$. Moreover, there is plenty of evidence in the literature that such an inclusion relation is strict. It was already noted by Eberhard \cite{eberhard}, and confirmed in \cite{liang} that certain Bell inequalities in the 2222 Bell scenario cannot be maximized if our only quantum resource are maximally entangled states. In the slightly more complex 3322 scenario ($|{\cal X}|=|{\cal Y}|=3$, $|{\cal A}|=|{\cal B}|=2$), Vidick \& Wehner \cite{vidick} proved that the $I_{3322}$ Bell inequality \cite{I3322}

\begin{eqnarray}
I_{3322}=&&-P_A(1|1)-P_B(1|0)-2P_B(1|1)+ \nonumber\\
&&+P(1,1|0,0)+P(1,1|0,1)+P(1,1|1,0)+P(1,1|1,1)- \nonumber\\
&&-P(1,1|0,2)+P(1,1|1,2)-P(1,1|2,0)+P(1,1|2,1)
\label{I3322_def}
\end{eqnarray}

\noindent can only be violated up to $I_{3322}=0.25$ using maximally entangled states. This number must be compared with the quantum maximum $I^{\Q}_{3322}\approx 0.250875$ \cite{I3322_PV}. Moreover, as shown by Palazuelos \& Junge \cite{junge}, as we increase the number of inputs and outputs, the difference between the projections of $\Q_{+}$ and $\Q$ becomes arbitrarily large.

So much for subquantum sets. How about supraquantum? Are there non-trivial instances of closed sets under wirings which contain $\Q$ strictly?

In \cite{mac_loc}, the authors identify the (closed) set of all distributions compatible with the principle of Macroscopic Locality with $\Q^1$, a first outer approximation to $\Q$ defined in \cite{quantum1}. $\Q^1$ can be shown to be different from the quantum set even in the 2222 scenario, where a CHSH violation of $2\sqrt{2}$ can be achieved with biased outcomes \cite{mac_loc}. It can also be seen that the maximal violation of the $I_{3322}$ inequality is $I_{3322}=0.3660$ (which number without truncation agrees with $(\sqrt3-1)/2$ up to eight decimal digits), much higher than its quantum counterpart.

However, besides this set and the no-signalling polytope, there are no more examples of postquantum sets in the literature. This leads us to wonder if, even though there exist infinitely many subquantum closed sets, there could be just a finite number of them containing $\Q$.

\begin{quest}
\label{supraquantum}
In a fixed non-locality scenario, how many different closed sets under wirings contain $\Q$?
\end{quest}

An interesting feature of the sets which we have examined so far is that they have a \emph{nested structure}, i.e., they satisfy $\L\subset \Q_\lambda \subset \Q_{+}\subset \Q\subset \Q^1\subset \NS$. One wonders if this is necessarily the case, i.e., if all conceivable physical theories form a hierarchy when only sets of correlations are considered.

\begin{quest}
\label{intersection}
Are there closed sets ${\cal S}_1$, ${\cal S}_2$ with non-trivial intersection ${\cal S}_1\cap {\cal S}_2\not={\cal S}_1,{\cal S}_2$?
\end{quest}

In order to motivate this question further, consider the following problem, faced in the axiomatization of quantum mechanics. Device-independent physical principles like Information Causality \cite{info_caus} or Macroscopic Locality \cite{mac_loc} have been proposed to limit the set of all reasonable physical distributions. The usual procedure to decide which distributions are ``compatible'' with these principles is to verify that, given $P(a,b|x,y)$, none of the distributions in the set $\S_{\{P(a,b|x,y)\}}$ violates such principles. However, if we denote by ${\cal Z}$ the set of all distributions satisfying some principle $Z$ in the previous sense, it could well be that

\be
\bigcup_{P\in {\cal Z}}\S_{\{P\}}
\ee

\noindent is not closed under wirings. That would imply that, even though all boxes in the set ${\cal Z}$ cannot violate principle $Z$ by themselves, some of them can be wired together into a pair of boxes that does not respect $Z$ anymore. This possibility is captured in the next definition:

\begin{defin}{\textbf{Stability under composition}}\\
A device-independent principle $Z$ is stable under composition iff, for any pair of closed sets $\S_1,\S_2$, compatible with $Z$, the set $\S_1+\S_2$ is also compatible with $Z$.
\end{defin}

\noindent Both Macroscopic Locality \cite{mac_loc} and the No-Signalling Principle \cite{PR_box} are stable under composition, and hence the notion of ``the largest set of boxes which satisfy Macroscopic Locality'' or ``the set of all boxes which respect No-signalling'' is well defined. On the other hand, expressions like ``the largest set of boxes which respect Information Causality'' may not have any meaning at all: indeed, if Information Causality is not stable under composition, there must be two or more maximal closed sets compatible with this principle. Finally, it is worth noticing that there are already examples of unstable physical principles in contextuality scenarios: in \cite{bin, barbara}, it is proven that no set of contextual correlations can \emph{strictly} contain $\Q_1$, as defined in \cite{tobias}, while satisfying Local Orthogonality (LO) \cite{tobias,cabello}. On the other hand, it is known that there are contextual correlations $q$ beyond $\Q_1$ which nevertheless satisfy LO. It hence follows that $q$, together with some of the correlations in $\Q_1$ can activate a violation of LO.

Coming back to nonlocality, note that any Bell-type inequality $B(P)=\sum_{a,b,x,y}B(a,b|x,y)P(a,b|x,y)\leq K_B$ can be interpreted as a bipartite device-independent physical principle. From this point of view, stability under composition and non-trivial intersection of closed sets are related by the next theorem:

\begin{theo}
\label{stable}
There exists a bipartite Bell-type linear inequality unstable under composition iff there exist two bipartite closed sets with non-trivial intersection.
\end{theo}

\begin{proof}
Suppose that bipartite closed sets of correlations have a nested structure, and let principle $Z$ be satisfied by both $\S_1$ and $\S_2$. The set $\S_1+ \S_2$ is generated by the boxes in $\S_1\cup \S_2$, which, by hypothesis, is equal to either $\S_1$ or $\S_2$, and hence is compatible with $Z$.

Conversely, suppose that there exist two sets $\S_1$, $\S_2$ with $\S_1\cap\S_2\not=\S_1,\S_2$. Let $P_1\in \S_1$, $P_1\not\in\S_2$ and $P_2\in\S_2$, $P_2\not\in \S_1$. Since $P_1\not\in \S_2$, by the Hahn-Banach theorem, there exists a Bell inequality $B_1$ such that

\be
B_1(P)\leq 0 \mbox{ for all } P\in \S_2, B_1(P_1)>0.
\ee

\noindent Likewise, there exists a Bell inequality $B_2$ with

\be
B_2(P)\leq 0 \mbox{ for all } P\in \S_1, B_2(P_2)>0.
\ee

\noindent W.l.o.g., let us assume that $B_1,B_2$ are defined in a Bell scenario with $|{\cal X}|=|{\cal Y}|=n$; that $\max_{P\in \S_1}B_1(P)=\max_{P\in \S_2}B_2(P)=1$; and that the respective maxima are attained by the distributions $P_1',P_2'$. Then, the $|{\cal X}|=|{\cal Y}|=2n$ Bell inequality

\begin{eqnarray}
(B_1\oplus B_2)(a,b|x,y)=&&B_1(a,b|x,y), \mbox{ for } x,y=1,...,n\nonumber\\
=&& B_2(a,b|x-n,y-n) \mbox{ for } x,y=n+1,...,2n\nonumber\\
=&&0, \mbox{ otherwise}
\end{eqnarray}

\noindent can be seen to satisfy $\max_{P\in \S_i}(B_1\oplus B_2)(P)\leq 1$ for $i=1,2$. Suppose, however, that Alice and Bob share a copy of $P'_1$ and a copy of $P'_2$, and they wire them together in such a way that, when asked $x=1,...,n$ ($y=1,...,n$), Alice (Bob) inputs $x$ ($y$) in the first box, and, when asked $x=n+1,...,2n$ ($y=n+1,...,2n$), she (he) inputs $x-n$ ($y-n$) in the second box. In either case, they output the result of the box they probe. Calling $(P'_1\oplus P'_2)(a,b|x,y)$ the resulting distribution, it can be verified that $(B_1\oplus B_2)(P_1'\oplus P_2')=2$. It follows that the the Bell-type inequality

\be
(B_1\oplus B_2)(P)\leq 1
\ee

\noindent is not stable under composition.

\end{proof}

\subsection{The tripartite case}
\label{tripartite_case}

We will now discuss how the notion of wirings changes when we move away from the bipartite Bell scenario. A tentative first definition could be:

\begin{defin}{\textbf{Physical closure (I)}}\\
A set $\S$ of tripartite correlations is closed if

\begin{enumerate}
\item $\S$ is convex.
\item {\cal S} is symmetric under permutations of the three parties.
\item For any collection of boxes $\{P_i(a,b,c|x,y,z)\}_{i=1}^n\subset {\cal S}$ and any triple of wirings ${\cal W}_A,{\cal W}_B,{\cal W}_C$, the box ${\cal W}_A\otimes {\cal W}_B\otimes {\cal W}_C(\otimes_{i=1}^nP_i)$ belongs to ${\cal S}$.
\end{enumerate}

\end{defin}

Unfortunately, this definition identifies as closed sets of boxes which cannot represent the correlations of any consistent physical theory. Consider the set ${\cal B}$ of all tripartite boxes $P(a,b,c|x,y,z)$ such that the bipartitions $P(a,b|x,y)$, $P(a,c|x,z)$, $P(b,c|y,z)$ are local. This set is clearly convex and symmetric. It is also closed under local wirings: indeed, given a collection of boxes $\{P_i(a,b,c|x,y,z)\}_{i=1}^n$ $\subset$ ${\cal B}$, by definition the bipartitions $\{P_i(a,b|x,y)\}_{i=1}^n$ belong to $\L$, and consequently, no matter how Alice and Bob wire them, they cannot violate locality. Likewise for Alice and Charlie and Bob and Charlie. After a local wiring, the resulting tripartite box will hence be bipartite local, and so it will belong to ${\cal B}$.

Now, consider the following tripartite generalization of the $PR$-box, called Box 44 in \cite{extreme}:

\begin{eqnarray}
PR_3(a,b,c|x,y,z)=&&\frac{1}{8}\delta_{a\oplus b\oplus c,xyz}.
\end{eqnarray}

\noindent $PR_3(a,b|x,y)=PR_3(a,c|x,z)=PR_3(b,c|y,z)=\frac{1}{4}$ for $x,y,z,a,b=0,1$. Such is a product distribution (and hence local): $PR_3(a,b,c|x,y,z)$ thus belongs to ${\cal B}$. Note, though, that, if Charlie inputs $z=1$, obtains the result $c=0$ and announces this to Alice and Bob, then the latter would be sharing a perfect (non-local) PR-box. Obviously, $PR(a,b|x,y)\delta_{c,0}\not\in {\cal B}$. We have just shown that ${\cal B}$ is not closed under \emph{post-selections}. Since post-selections are physically legitimate operations, ${\cal B}$ cannot be a physical set.

Similarly, suppose that no post-selections are made, but, at the beginning of the experiment, $PR_3(a,b,c|x,y,z)$ is distributed in such a way that Alice receives the first box; and Bob, the second and third. Then Bob can perform the wiring indicated in Figure \ref{wiring_tri} in order to engineer a perfect PR-box between him and Alice. We hence conclude that ${\cal B}$ is also not closed under \emph{distribution and wirings}.

\begin{figure}
  \centering
  \includegraphics[width=8.5 cm]{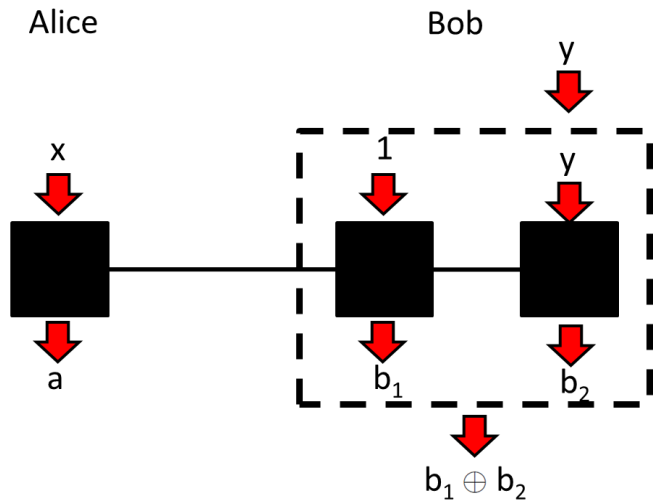}
  \caption{\textbf{Distribution and wirings.} If Bob is distributed, not just one, but two of the boxes corresponding to the triple $PR_3(a,b,c|x,y,z)$, he can wire them together to produce perfect PR-box correlations with Alice.}
  \label{wiring_tri}
\end{figure}

Any Bell experiment has a preparation stage and a measurement stage. The post-selection and distribution of boxes corresponds to the former; local wirings, to the latter. All these are valid physical operations, and the multipartite definition of closure must account for them.

\begin{defin}{\textbf{Physical closure (II)}}\\
\label{closure_tri}
A set $\S$ of tripartite correlations is closed if

\begin{enumerate}
\item $\S$ is convex.
\item ${\cal S}$ is closed under post-selections.
\item For any collection of boxes $\{P_i(a,b,c|x,y,z)\}_{i=1}^n\subset {\cal S}$, \emph{distributed arbitrarily to the three parties}, and any triple of wirings ${\cal W}_A,{\cal W}_B,{\cal W}_C$, the box ${\cal W}_A\otimes {\cal W}_B\otimes {\cal W}_C(\otimes_{i=1}^nP_i)$ belongs to ${\cal S}$.
\end{enumerate}

\end{defin}
\noindent We will be using this definition from now on.

One wonders if there exist at all non-trivial sets with genuine tripartite nonlocality besides $\Q$ and $\NS$. If such were not the case, then the existence of genuine tripartite nonlocality could already be interpreted as a physical principle to single out the quantum set. Regrettably, life is not that easy. Have a look at the following definition.

\begin{defin}{\textbf{The Time Ordered BiLocal set $\TOBL$}}\\
A tripartite distribution $P(a_1,a_2,a_3|x_1,x_2,x_3)$ belongs to $\TOBL$ iff it admits the Time Ordered BiLocal (TOBL) expansion \cite{extreme,TOBL}, i.e., if it admits a decomposition of the form:

\begin{eqnarray}
\label{eq:tobl}
  &P(a_1, a_2, a_3|x_1, x_2, x_3)=\nonumber\\
  &= \sum_\lambda p_\lambda^{i|jk} P(a_i|x_i, \lambda) P_{j\rightarrow k}(a_j, a_k|x_j, x_k,\lambda) \nonumber\\
  &= \sum_\lambda p_\lambda^{i|jk} P(a_i|x_i, \lambda) P_{j\leftarrow k}(a_j, a_k|x_j, x_k,\lambda)
\end{eqnarray}
for $(i,j,k)= (1,2,3), (2,3,1),(3,1,2)$, with the distributions
$P_{j\rightarrow k}$ and $P_{j\leftarrow k}$ obeying the
conditions
\begin{eqnarray}
 &P_{j \rightarrow k}(a_j|x_j,\lambda) = \sum_{a_k} P_{j \rightarrow k}(a_j, a_k|x_j, x_k,\lambda),\nonumber\\
 &P_{j \leftarrow k}(a_k|x_k,\lambda) = \sum_{a_j} P_{j \leftarrow k}(a_j, a_k|x_j, x_k,\lambda).
\end{eqnarray}

\end{defin}

As shown in \cite{TOBL}, this polytope is closed. Moreover, by definition, no distribution in $\TOBL$ can exhibit bipartite non-locality, and hence there are points in $\Q$ which are not present in $\TOBL$. On the other hand, some boxes in $\TOBL$ violate the Guess Your Neighbour's Input inequality \cite{gyni}, which cannot be violated by quantum mechanical systems. We thus have that the existence of $\TOBL$ answers affirmatively the multipartite versions of Question \ref{polytope} (existence of closed non-trivial polytopes) and Question \ref{intersection} (existence of closed sets with non-trivial intersection).

Sadly, the restriction of $\TOBL$ to two parties coincides with $\L$, and so the original -bipartite- questions remain open.

Another non-trivial tripartite closed set is the set of correlations generated by conducting measurements on quantum states invariant under partial transposition \cite{PPT} with respect to parties $A$, $B$, and $C$ (i.e., with respect to any bipartition of the three parties). For this type of states, no bipartite entanglement can be distilled between any groups of parties. Nevertheless, as it has been shown in \cite{ppt3}, there exist such quantum states which violate a tripartite Bell inequality (number 5) of the list of Sliwa~\cite{sliwa}. Let us designate the corresponding set by $\Q_{PPT}$. As one can check, the correlations arising from $\Q_{PPT}$ fulfill all conditions of Definition~\ref{closure_tri}, hence the set is closed under wirings.

We were not able to find more examples of non-trivial tripartite closed sets in the literature, although clearly another set can be generated by intersecting $\TOBL$ with $\Q$. The goal of reconstructing $\Q$ via device-independent principles motivates our next question:

\begin{quest}
\label{tri_ex}
Are there (non-trivial) tripartite supraquantum closed sets?
\end{quest}

\section{A zoo of computable closed sets}
\label{zoo}
In this section we will provide a collection of closed sets of correlations which admit an efficient characterization. We will divide them into two: those which can be characterized via linear programming \cite{linear_prog} (polytopes), and those for which the more sophisticated semidefinite programming (SDP) tools \cite{sdp} are necessary.

\subsection{Polytopes}

In the previous section we gave an example of a non-trivial tripartite closed polytope. It is certainly ironic that, in order to construct closed bipartite polytopes, we must turn to multipartite notions again. Consider the following definition:

\begin{defin}{\textbf{Ghost World}}\\
\label{ghost}
A bipartite distribution $P(a,b|x,y)$ belongs to Ghost World (denoted $\GW$) iff there exist tripartite distributions $P_A(a,a',b|x,x',y)$, $P_B(a,b,b'|x,y,y')\in \NS$ such that

\begin{enumerate}
\item $P_A(a,a',b|x,x',y)$ ($P_B(a,b,b'|x,y,y')$) is invariant under the exchange of the first two systems (the last two systems).
\item $P(a,b|x,y)=\sum_{a'}P_A(a,a',b|x,x',y)=\sum_{b'}P_B(a,b,b'|x,y,y')$.

\end{enumerate}

\end{defin}

A way to picture $\GW$ is to imagine that each party $A$ or $B$ has an associated ``ghost'', $A'$ or $B'$, which has similar experiences. Of course, ghosts do not exist, so there are no joint four-partite distributions for $AA'BB'$.

\begin{theo}

$\GW$ is closed under wirings.

\end{theo}

\begin{proof}
The proof is sketched in Figure \ref{ghost_fig}. Call $P_f(a,b|x,y)$ the final distribution, and note that, for each pair of boxes $P_i(a,b|,xy)$ involved in the wiring, there is a tripartite extension $P_i(a,a',b|x,x',y)$ satisfying the conditions of Definition \ref{ghost}. No matter how complicated Alice's wiring is, Alice's ghost can always mimic it, and hence there exists a tripartite distribution $P_f(a,a',b|x,x',y)$, with $P_f(a,a',b|x,x',y)=P_f(a',a,b|x',x,y)$ and $\sum_{a'}P_f(a,a',b|x,x',y)=P_f(a,b|x,y)$. Likewise, $P_f(a,b|x,y)$ admits a symmetric extension of Bob's part, and hence $P_f(a,b|x,y)\in \GW$.

\end{proof}

\begin{figure}
  \centering
  \includegraphics[width=8.5 cm]{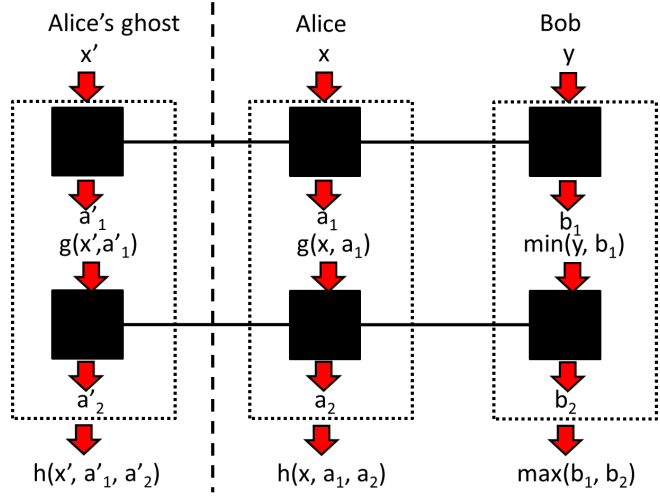}
  \caption{\textbf{Closure of Ghost World.} Any wiring of Alice's can be mimicked by her ghost, thus giving rise to a tripartite distribution symmetric under the exchange of Alice and her ghost.}
  \label{ghost_fig}
\end{figure}

\noindent Clearly, one can define analogous closed sets $\GW_n$ by demanding that Alice and Bob admit, not just one, but $n$ ghosts each. All such sets are polytopes strictly contained in $\NS$: this can be seen from the fact that, due to monogamy \cite{lluis}, no distribution in $\GW_n$ can violate  Bell inequalities with $|{\cal X}|=n+1$.

Note that the sets $\GW_n$ do not admit a straightforward extension to the multipartite case, since they are not closed under post-selection or distribution and wirings. It is easy, though, to modify Definition \ref{ghost} so that they are.

\begin{defin}{\textbf{Twin World}}\\
\label{twin}
A tripartite distribution $P(a,b,c|x,y,z)$ belongs to Twin World (denoted $\TW$) iff there exists an $AA'BB'CC'$-partite distribution $P(a,a',b,b',c,c'|x,x',y,y',z,z')$ $\in \NS$ such that

\begin{enumerate}
\item $P(a,a',b,b',c,c'|x,x',y,y',z,z')$ is invariant under the exchange of $A$ and $A'$; $B$ and $B'$; and $C$ and $C'$.
\item $P(a,b,c|x,y,z)=\sum_{a',b',c'}P(a,a',b,b',c,c'|x,x',y,y',z,z')$.

\end{enumerate}

\end{defin}

\noindent Now the intuition is that each inhabitant of Twin World has a twin brother or sister who, contrary to ghosts, simultaneously co-exist.

The proof of the closure of Twin World is analogous to that of Ghost World, and hence it will not be repeated. As before, one can define new sets $\TW_n$ by imposing that each party has $n$ twins, instead of just one.

\subsection{Shadows of spectrahedra}

In \cite{quantum2}, the authors define a sequence of SDP approximations $\Q^n$ to the bipartite quantum set $\Q$ with the property $\Q^1\supset \Q^2\supset...\supset \Q$. In the following we will prove that each of these sets is closed under wirings, as well as some intermediate sets also considered in the literature.

Assign ``symbolic projectors'' $\{E^x_a,F^y_b\}$ to any possible local event $(x,a)$ or $(y,b)$. We will call $A$ and $B$ the set of projectors associated to Alice's and Bob's interactions. By $A^mB^n$, we will denote the set of all sequences of products of projectors of the form $E^{x_1}_{a_1}E^{x_2}_{a_2}...E^{x_{m'}}_{a_{m'}}F^{y_1}_{b_1}F^{y_2}_{b_2}...F^{y_{n'}}_{b_{n'}}$, where $m'\leq m$, $n'\leq n$. Note that, if two consecutive projectors correspond to different outcomes of the same measurement, the corresponding sequence will be equal to 0. Also, taking $m'=n'=0$, we have that $\id\in A^mB^n$, for all $m,n$. We will assume that these symbolic operators can be conjugated and multiplied at a formal level, with the peculiarity that all the $F$'s commute with all the $E$'s.

\begin{defin}{${\cal O}$-positivity}\\
\label{O_posit}
Let ${\cal O}$ be some (symmetric) union of sets of operators of the type $A^mB^n$ containing $A$ and $B$, and denote by $\overline{{\cal O}}$ its associated (real) vector space, i.e., the set of all linear combinations of the elements of ${\cal O}$. Likewise, denote by $\overline{{\cal O}}^2$ the set of linear combinations of the elements of ${\cal O}\cdot {\cal O}$ (products of two elements of ${\cal O}$). We will say that Alice's and Bob's no-signalling shared probability distribution $P(a,b|x,y)$ is ${\cal O}$-positive if there exists a linear functional $L:\overline{{\cal O}}^2\to \R$ such that

\begin{enumerate}
\item
\label{primer}
$L(1)=1$.
\item
$L(ff^\dagger)\geq 0$ for any $f\in\overline{{\cal O}}$.
\item
\label{ulti}
$P(a,b|x,y)=L(E^x_a\cdot F^y_b)$.
\end{enumerate}
\end{defin}

For any ${\cal O}$, the set $\Q^{\cal O}$ of ${\cal O}$-positive distributions is convex and contains $\Q$. Moreover, deciding ${\cal O}$-positivity can be formulated as a semidefinite program whose complexity scales polynomially with the number of measurement settings\footnote{Indeed, notice that the constraint $L(ff^\dagger)\geq 0$ is equivalent to demanding that the matrix $\Gamma$, defined as $\Gamma_{s,t}\equiv L(S^\dagger T)$, is positive semidefinite.}. Actually, note that, if ${\cal O}=\bigcup_{m=0}^N A^mB^{N-m}$, then $\Q^{\cal O}= \Q^N$, with $\Q^N$ defined as in \cite{quantum2}.

\begin{theo}
\label{equivalence}
For any ${\cal O}$, a given set of bipartite correlations $P(a,b|x,y)$ is ${\cal O}$-positive iff there exist projectors $\{E^x_a,F^y_b\}$ and a quantum state $\rho$ such that

\begin{enumerate}
\item
$\tr(\rho S^\dagger T)=\tr(\rho U^\dagger V)$, if $s^\dagger t=u^\dagger v$, where $S,T,U,V$ are sequences of products of the projectors $\{E^x_a,F^y_b\}$ corresponding to the elements $s,t,u,v\in {\cal O}$.

\item
$E^x_aE^x_{a'}=F^y_bF^y_{b'}=0$ for $a\not=a',b\not=b'$.

\item
$\sum_{a}E^x_a=\sum_{b}F^y_b=\id$.
\item
$P(a,b|x,y)=\tr(\rho E^x_a F^y_b)$,
\end{enumerate}

\end{theo}

\noindent For a proof, see \cite{quantum2}.

We will next prove this section's main theorem.

\begin{theo}
\label{closed_hierarchy}
Let ${\cal O}$ be the union of some of the sets $A^mB^n$, including $A$ and $B$. Then, $\Q^{\cal O}$ is closed under wirings.
\end{theo}

\begin{proof}

Any general wiring of different boxes $\{P_i(a,b|x,y)\}_{i=1}^N$ is the result of composing two operations:

\begin{enumerate}
\item
A \emph{complete deterministic measurement strategy}, that starts by performing some measurement $X$ on some box and continues through the sequential measurement of \emph{all} the remaining boxes, where which box to measure and which interaction to apply to it are a function of the previous measurement outcomes.
\item
The identification of some of the different outcomes of the former measurement strategies.
\end{enumerate}

In order to prove that a certain set of correlations containing the local set is closed under wirings, it is thus enough to show that it is closed under the previous two operations. The next two lemmas do the job.

\begin{lemma}
Let $P(a,b|x,y)$ be ${\cal O}$-positive for some set ${\cal O}$. Then, any identification $P'(a,b|x,y)$ of the outcomes of $P(a,b|x,y)$ is also ${\cal O}$-positive.
\end{lemma}

\begin{proof}

It is enough to prove that, for some input $\bar{x}$, the identification of any two different outcomes generates a new ${\cal O}$-positive distribution. Take thus two outputs $\tilde{a}\not=\tilde{a}'$, rename them $\bar{a}$, and consider the reduced distribution $P'(a,b|x,y)$, with $P'(\bar{a},b|\bar{x},y)=P(\tilde{a},b|\bar{x},y)+P(\tilde{a}',b|\bar{x},y)$ and $P'(a,b|x,y)=P(a,b|x,y)$, for $a\not=\tilde{a},\tilde{a}'$ or $x\not=\bar{x}$. By Theorem \ref{equivalence}, we have to show that there exists a quantum state and a set of measurements satisfying conditions 1-4.

Easy: let $\rho, \{E^x_a,F^y_b\}$ be the quantum states and operators associated to $P(a,b|x,y)$. Now, take $\rho'=\rho$, $(F^y_{b})'=F^y_b$, $(E^x_a)'=E^x_a$, for $a\not=\tilde{a},\tilde{a}'$ or $x\not=\bar{x}$, and $(E^{\bar{x}}_{\bar{a}})'=E^{\bar{x}}_{\tilde{a}}+E^{\bar{x}}_{\tilde{a}'}$. Then it is clear that the state $\rho'$ and the operators $\{(E^x_a)',(F^y_b)'\}$ satisfy the conditions of Theorem \ref{equivalence} for $P'(a,b|x,y)$.

\end{proof}

\begin{lemma}
Given a set of operators ${\cal O}$, let $\{P_i(a,b|x,y)\}_{i=1}^N$ be any set of ${\cal O}$-positive boxes. Then, any complete deterministic wiring of them is ${\cal O}$-positive.

\end{lemma}

\begin{proof}

Consider a measuring strategy $\bar{x}$, one of whose possible outcomes is $\vec{a}=(i_1,x_1,a_1)\to (i_2,x_2,a_2)\to ...\to (i_N,x_N,a_N)$, understood as the outcome corresponding to measuring $x_1$ and obtaining outcome $a_1$ in box $i_1$, followed by a measurement $x_2$ that outputs $a_2$ in box $i_2$, etc. Call $a(\vec{a},i)$ the output in $\vec{a}$ corresponding to the box $i$; and $x(\vec{a},i)$, the corresponding measurement $x_i$. Given the representation $\{\rho^i,E^{i,x}_a,F^{i,y}_b\}_{i=1}^N$ of the boxes $\{P_i\}_{i=1}^N$, define the normalized quantum state $\rho\equiv\otimes_{i=1}^N\rho^i$ and the projector operators

\begin{eqnarray}
&&E^{\bar{x}}_{\vec{a}}\equiv\otimes_{i=1}^NE^{i,x(\vec{a},i)}_{a(\vec{a},i)},\nonumber\\
&&F^{\bar{y}}_{\vec{b}}\equiv\otimes_{i=1}^NF^{i,y(\vec{b},i)}_{b(\vec{b},i)}.
\end{eqnarray}

Then it is immediate to check that

\be
P(\vec{a},\vec{b}|\bar{x},\bar{y})=\prod_{i=1}^NP_i(a(\vec{a},i),b(\vec{b},i)|x(\vec{a},i),y(\vec{b},i))=\tr(\rho E^{\bar{x}}_{\vec{a}}E^{\bar{y}}_{\vec{b}}),
\ee
\noindent so condition 4 of Theorem \ref{equivalence} is satisfied. Now, suppose that $\vec{a},\vec{a}'$ correspond to different outcomes of the same deterministic strategy $\bar{x}$. This implies that, at some point in the measuring process, some measurement $x$ was performed on some box $k$ that output different results for $\vec{a},\vec{a}'$, i.e., $a(\vec{a},k)\not=a(\vec{a}',k)$, and so $E^{k,x(\vec{a},k)}_{a(\vec{a},k)}\cdot E^{k,x(\vec{a}',k)}_{a(\vec{a}',k)}=0$. Thus $E^{\bar{x}}_{\vec{a}}E^{\bar{x}}_{\vec{a}'}=0$ for any pair of sequential outcomes $\vec{a}\not=\vec{a}'$ corresponding to the same measurement strategy $\bar{x}$. The same holds for the $F$'s, and so condition 2 of Theorem \ref{equivalence} is likewise respected.

Let $\bar{x}$ be any deterministic (possibly incomplete) measurement strategy. For any outcome $\vec{a}$ of $\bar{x}$, define $E^{x(\vec{a},i)}_{a(\vec{a},i)}=\id_i$ in case the box $i$ was not measured, and denote by $E^{\bar{x}}_{\vec{a}}$ the projector $E^{\bar{x}}_{\vec{a}}=\otimes_{i=1}^NE^{x(\vec{a},i)}_{a(\vec{a},i)}$. We will next prove that $\sum_{\vec{a}} E^{\bar{x}}_{\vec{a}}=\id$. The proof will proceed by induction on the maximum length $n$ of the sequences of outcomes. If $n=1$, $\sum_{\vec{a}} E^{\bar{x}}_{\vec{a}}=\sum_{a}E^{x(\vec{a},i)}_a\otimes \id_{1,...,i-1,i+1,...n}$, for some measurement $x$ over box $i$. By hypothesis the last sum must then be equal to the identity. Now, let us assume that the property holds for strategies of $n$ consecutive measurements, and take a measurement strategy $\bar{x}$ of $n+1$ consecutive measurements. Then we can always decompose $\bar{x}$ as some strategy $\bar{x}'$ (where no more than $n$ measurements have to be made in order to announce an outcome) plus a new measurement $x(\vec{a}')$ in box $k(\vec{a}')$ depending on the possible outcomes $\vec{a}'$ of $\bar{x}'$. Note that, for some possible outputs $\vec{a}'$ of $\bar{x}$, $x(\vec{a}')$ can be trivial. Thus we have that

\begin{eqnarray}
&&\sum_{\vec{a}}E^{\bar{x}}_{\vec{a}}=\sum_{\vec{a}'}\sum_{a} E^{\bar{x}'}_{\vec{a}'}\otimes E^{k(\vec{a}'),x(\vec{a}')}_a=\nonumber\\
&&=\sum_{\vec{a}'}E^{\bar{x}'}_{\vec{a}'}\otimes \id_{k(\vec{a}')}=\id.
\end{eqnarray}

\noindent Condition 3 of Theorem \ref{equivalence} is therefore fulfilled by complete deterministic strategies.

Finally, take any operator identity $s^\dagger t=u^\dagger v$, for $s,t,u,v\in {\cal O}$. Then,

\begin{eqnarray}
\tr(\rho S^\dagger T)&&=\prod_{i=1}^N \tr(\rho^i (S^i)^\dagger T^i)=\prod_{i=1}^N \tr(\rho^i (U^i)^\dagger V^i)=\nonumber\\
&&=\tr(\rho U^\dagger V).
\end{eqnarray}

\noindent Hence condition 1 of Theorem \ref{equivalence} holds as well. It follows that $P(\vec{a},\vec{b}|\bar{x},\bar{y})$ is ${\cal O}$-positive.

\end{proof}

\end{proof}

This closure result can be easily extended to other bipartite sets of correlations defined via SDP hierarchies where certain extra restrictions on $\Gamma$ are imposed. For instance, let $s_A,t_A$ ($s_B,t_B$) denote sequences of Alice's (Bob's) operators, and, consider the constraints

\be
L(S_AS_B)=L(S^\dagger_AS_B), L(S_AS_B)=L(S_AS^\dagger_B),
\label{extra}
\ee

\noindent invoked by Moroder \emph{et al.} \cite{moroder} for the characterization of the correlations attainable by quantum states invariant under partial transposition \cite{PPT}. For any set of operators ${\cal O}$, the set of all correlations in $\Q^{\cal O}$ satisfying the above restriction will be denoted as $\Q^{\cal O}_{PPT}$. It is clear that, for all sets ${\cal O}$ satisfying the conditions of Theorem \ref{closed_hierarchy}, $\Q^{\cal O}_{PPT}$ is closed under wirings, since the above relations are ``inherited'' by the measurement operators corresponding to the wirings. Let ${\cal O} = A^1B^1$, and consider the $I_{3322}$ inequality~(\ref{I3322_def}). In this case, we got the bound 0.15 (which is lower than the quantum maximum~$\approx 0.250875$), proving that the set $\Q^{A^1B^1}_{PPT}$ is non-trivial even in the simplest meaningful ternary-input binary-output setting.

In the multipartite case, one can as well define sets of operators for, say, Alice, Bob and Charlie. Consider those of the form ${\cal O}_N\equiv A^NB^NC^N$, whose associated sets $\tilde{{\cal Q}}^N\equiv \Q^{{\cal O}_N}$ were proposed in \cite{graph} to characterize quantum contextuality. In \cite{almost}, it is shown that $\tilde{{\cal Q}}^1$, the ``almost quantum'', set is closed in the multipartite sense specified by definition \ref{closure_tri}. It is easy to adapt the proof to the case of general $N$. Let us also recall the tripartite $\Q_{PPT}$ set introduced in Sec.~\ref{tripartite_case}. The counterpart of this set associated to the hierarchy on any level $N$ defines a valid closed set as well. The fact that this set is distinct from $\Q_{PPT}$ for some $N$ is indicated by the result of Ref.~\cite{moroder} proving that at least level $N=3$ of the hierarchy is required to recover the maximum $\Q_{PPT}$ value of inequality number 5 from the list of Ref.~\cite{sliwa}.

In Section \ref{questions} we considered the set $\Q_{+}$ of all quantum distributions attainable with maximally entangled quantum states. We did not explain, however, how to characterize this set. We will conclude by introducing a hierarchy of semidefinite programs to bound $\Q_{+}$. As in the previous examples, all such relaxations define closed sets of correlations, and therefore constitute useful tools to explore nonlocality distillation.

Suppose then that Alice and Bob share arbitrarily many copies of the maximally entangled state, or, equivalently, a $d$-dimensional maximally entangled state $\ket{\Psi_d}=\frac{1}{\sqrt{d}}\sum_{j=1}^d\ket{j}\otimes\ket{j}$. They probe these states with the Positive Operator Valued Measures $E^x_a\geq 0, F^y_b\ge 0$, with $\sum_aE^x_a=\sum_bF^y_b=\id$. Then it can be verified that

\begin{eqnarray}
&&P(a,b|x,y)=\bra{\Psi}E^x_a\otimes F^y_b\ket{\Psi}=\frac{1}{d}\tr(E^x_a(F^y_b)^T),\nonumber\\
&&P(a|x)=\bra{\Psi}E^x_a\otimes \id_d\ket{\Psi}=\frac{1}{d}\tr(E^x_a),\nonumber\\
&&P(b|y)=\bra{\Psi}\id_d\otimes F^y_b\ket{\Psi}=\frac{1}{d}\tr((F^y_b)^T),
\end{eqnarray}

\noindent with the symbol $\bullet^T$ denoting transposition with respect to the basis  $\{\ket{j}\}_{j=1}^d$.

We relax the previous relation to $P(a,b|x,y)=\Tr(E^x_aF^y_b)$, where $\Tr$ denotes a normalized tracial state, i.e., $\Tr(\id)=1$, $\Tr(AB)=\Tr(BA)$ for all operators $A,B$. Note that we have redefined Bob's operators as $F^y_b\to (F^y_b)^T$. This problem can be attacked using a modified version of Burgdorf \& Klep's hierarchy of SDPs \cite{klep}. We arrive at the following definition:

\begin{defin}
\label{Qplus_hier}
Let ${\cal P}_k$ be the set of all monomials of the symbolic operators $E^x_a,F^y_b$, with $E^x_aF^y_b\not=F^y_bE^x_a$ and $\sum_aE^x_a=\sum_b F^y_b=\id$, up to degree $k$. A bipartite distribution $P(a,b|x,y)$ belongs to $\Q^{k}_{+}$ iff there exists a linear functional $L:\overline{{\cal P}}^2_k\to \R$ such that

\begin{enumerate}
\item
\label{primer2}
$L(1)=1$.
\item
\label{positive1}
$L(ff^\dagger)\geq 0$ for any $f\in\overline{{\cal P}}_k$.

\item
\label{positive2}
$L(fE^x_af^\dagger),L(fE^x_af^\dagger F^y_b)\geq 0$ for any $f\in\overline{{\cal P}}_{k-1}$.
\item
\label{positive3}
$L(fE^x_af^\dagger E^{x'}_{a'}),L(fF^y_bf^\dagger F^{y'}_{b'}),L(fE^x_af^\dagger F^{y}_{b})\geq 0$ for any $f\in\overline{{\cal P}}_{k-1}$.
\item
\label{tracial}
$L(ST)=L(TS)$, for $ST\in $.
\item
\label{ulti2}
$P(a,b|x,y)=L(E^x_a,F^y_b)$.
\end{enumerate}
\end{defin}

\noindent The tracial nature of the underlying quantum state is expressed in condition \ref{tracial} by imposing that averages of products of operators must be invariant under cyclic permutations, and that expressions of the form $L(fXf^\dagger Y)$ are non-negative for positive semidefinite $X$, $Y$. As before, imposing conditions \ref{positive1}, \ref{positive2}, \ref{positive3} amounts to verifying that certain moment-like matrices are positive semidefinite, as in \cite{siam}. Note that, since $E^x_a$ and $F^y_b$ act on the same space, they do not commute any longer, unlike in Definition \ref{O_posit}. At a practical level, this implies that the decidability of $\Q_{+}^N$ is harder than that of $\Q^N$, as the former amounts to verifying the existence of matrices with many more rows and columns.

The proof of the closure of $\Q^N_{+}$ is very similar to that of $\Q^N$. Here follows a sketch: first, one must prove that $\Q^N_{+}$ is closed under identification of outcomes. This is not difficult, since the positivity conditions \ref{positive1}, \ref{positive2}, \ref{positive3} of definition \ref{Qplus_hier} are still satisfied when we perform the assignments $E^x_{\bar{a}}\to E^x_a+E^x_{a'}$. The next step is to prove that, under complete measurement strategies over $n$ boxes in $\Q^N_{+}$, there also exist linear functionals satisfying conditions \ref{positive1}, \ref{positive2}, \ref{positive3} for the resulting distribution. Taking such a functional to be $L(\prod_{i=1}^nS^i)=\prod_{i=1}^nL_i(S^i)$, conditions \ref{positive1}, \ref{positive2}, \ref{positive3} are a consequence of the fact that the tensor product of positive semidefinite matrices is also positive semidefinite.

\section{Some answers from the zoo}
\label{answers}

The zoo of closed sets allows to answer several of the questions posed in Section \ref{questions}.

We already saw that, due to monogamy \cite{lluis}, in all non-trivial Bell scenarios, $\GW_n\not=\NS$. Hence, if we could prove that $\GW_n\not=\L$, we could answer Question \ref{polytope} on the existence of non-trivial closed polytopes.

It is worth noticing that, under the hypothesis that $P\not=NP$, for every $n$ there must be some Bell scenario where $\GW_n$ contains $\L$ strictly. Indeed, on one hand linear optimizations over $\L$ constitute an NP-hard problem \cite{pitowsky}. On the other hand, for fixed $n$, $\GW_n$ can be described in each Bell scenario via a linear program that scales polynomially with the number of measurement settings. If $\GW_n=\L$, it would thus be possible to characterize $\L$ with a polynomial algorithm, contradicting the NP-hardness of characterizing $\L$. Similar considerations apply to the related polytopes $\TW_n$, also defined in Section \ref{zoo}.

This answers Question \ref{polytope}, if only under popular conjectures in complexity theory. We can get rid of such conjectures, however, precisely by exploiting the fact that all such sets admit an efficient characterization. Due to monogamy, $\GW=\L$ in the 2222 case; we must therefore explore more complex Bell scenarios.

The local regions of the 3322 scenario are defined by the CHSH Bell inequality \cite{chsh} and the $I_{3322}$ inequality \cite{I3322}, see eq. (\ref{I3322_def}). Using the MATLAB packages YALMIP \cite{yalmip} and SeDuMi \cite{sedumi}, we maximized $I_{3322}$ over all distributions in $\GW$, obtaining the value $1/3$, achievable by the box given in the table below.

\begin{center}
\begin{tabular}{|p{4cm}|p{1cm}|p{1cm}|p{1cm}|p{1cm}|}
\hline
 $P(a,b|x,y)$ & \multicolumn{4}{c|}{Outputs ($ab$)} \\ \hline
Inputs ($xy$) & 00 & 01 & 10 & 11 \\ \hline
00, 01, 10 and 11 & 1/3 & 0 & 1/3 & 1/3 \\ \hline
02 & 0 & 1/3 & 1/2 & 1/6 \\ \hline
12 & 1/3 & 0 & 1/6 & 1/2 \\ \hline
20 & 1/6 & 1/3 & 1/2 & 0 \\ \hline
21 & 1/2 & 0 & 1/6 & 1/3 \\ \hline
22 & 1/4 & 1/4 & 1/4 & 1/4 \\ \hline
\end{tabular}
\end{center}

\noindent The value $I^{\GW}_{3322}=\frac{1}{3}$ is quite far from both the local ($I^{\L}_{3322}=0$) and no-signalling limits ($I^{\NS}_{3322}=1$): $\GW$ is thus enough to answer Question \ref{polytope}. This result also answers Question \ref{intersection}, on the existence of closed sets with non-trivial intersection. Indeed, as we already noted, $\GW$ cannot violate the 2-setting CHSH inequality \cite{chsh}, and so there exists a point in $\Q$ not contained in $\GW$. On the other hand, as we just saw, $\GW$ violates $I_{3322}$ by an amount significantly greater than the quantum maximum of $I_{3322}^Q\approx 0.250875$ \cite{I3322_PV}; ergo, there are points in $\GW$ which are separated from $\Q$. $\Q\cap \GW$ is therefore neither $\Q$ nor $\GW$, and thus, by Theorem \ref{stable}, there exist bipartite Bell inequalities, i.e., physical principles (satisfied by quantum theory!), which are not stable under composition.

In the nonlocality community there is a generalized conjecture that, even for the 2222 scenario, the sets $\{\Q^N\}_N$ are all different from each other. This implies (under this assumption) that there exist infinitely many supraquantum closed sets, all of which differ in the same Bell scenario, hence answering Question \ref{supraquantum}. Under the conjecture $P\not=NP$ and Kirchberg's conjecture \cite{scholz, connes1,connes2}, we can prove a weaker result:

\begin{lemma}
Assume that $P\not=NP$ and that Kirchberg's conjecture holds. Then, for any $k\in \N$, there exist $k$ numbers $N_1,...,N_k$ and a Bell scenario with $|{\cal A}|=|{\cal B}|=2$ where the sets $\{Q^N_i\}_{i=1}^k$ are all different.

\end{lemma}

\begin{proof}
We will prove it by \emph{reductio ad absurdum}. Assume that, for any Bell scenario with $|{\cal A}|=|{\cal B}|=2$, there are no more than $k$ different sets of the type $\{Q^N_i\}_i$, and that this number is tight. Then, there must exist $s\in \N$ such that, in the Bell scenario $|{\cal X}|=|{\cal Y}|=s$, there exist $N_1<...<N_k\in \N$, with $\Q^{N_i}\not=\Q^{N_j}$, for $i\not=j$. Now, characterizing $\Q$ in Bell scenarios with $|{\cal A}|=|{\cal B}|=2$ is an NP-hard problem in the number of measurement settings \cite{QBell_NP}. Since $\Q^{N_k}$ admits an efficient semidefinite programming description, by $P\not=NP$ we thus have that, for some $n>s$, $\Q^{N_k}\not=\Q$ in the $n$-setting Bell scenario. Also, by Kirchberg's conjecture we have that $\lim_{N\to\infty}\Q^N=\Q$. It follows that there exists $N_{k+1}>N_k$ such that $\Q^{N_{k+1}}\subsetneq \Q^{N_{k}}$ in the $n$-setting Bell scenario. In that scenario, the sets $\Q^{N_1},...,\Q^{N_k}$ must be all different, since their projections to $s$-setting Bell scenarios are. We therefore conclude that $\Q^{N_{k+1}}$ is different from $\Q^{N_1},...,\Q^{N_k}$, hence contradicting the initial hypothesis.

\end{proof}

Moving to the tripartite setting, we saw that each of the sets $\{\Q^{A^NB^NC^N}\}_N$ constitutes a tripartite supraquantum closed set. That some of these sets are non-trivial can be seen from the fact that there exist tripartite Svetlichny-type inequalities for which correlations within the set $\Q^{A^1B^1C^1}$ does not return the quantum bound. Indeed, let us pick inequality 409 from the list of Bancal \emph{et al.}~\cite{bancal}:
\begin{eqnarray}
S_{409}=&&  -2P_A(0|1) - 2 P_{A,B}(0,0|1,2) + 2 P_{A,B}(0,0|2,2)  - P(0,0,0|1,1,1)+ \nonumber\\
        && + 4P(0,0,0|1,1,2) + 2P(0,0,0|1,2,2) - 2P(0,0,0|2,2,2) + \nonumber\\
        &&+\text{sym} \leq 0,
\label{S409}
\end{eqnarray}

\noindent where sym denotes the missing permutationally invariant terms. The maximum within the set $\Q^{A^1B^1C^1}$ is $0.0221$, however by increasing the level of hierarchy we eventually recover the quantum maximum of $0.0132$ up to numerical precision.

This answers Question \ref{tri_ex} in the affirmative.

Finally, using the sets defined in the previous section, one can find points of non-distillability for different Bell inequalities. As before, optimizations were conducted using the MATLAB packages YALMIP \cite{yalmip} and SeDuMi \cite{sedumi}. The results, together with an indication of which sets we used to derive them, appear in Figures \ref{CGLMP_fig} and \ref{I3322_fig}. Each line in the graphs indicates the existence of a pair of boxes with the property that arbitrarily many copies of them do not allow to increase the indicated Bell violation.

\begin{figure}
  \centering
  \includegraphics[width=6 cm]{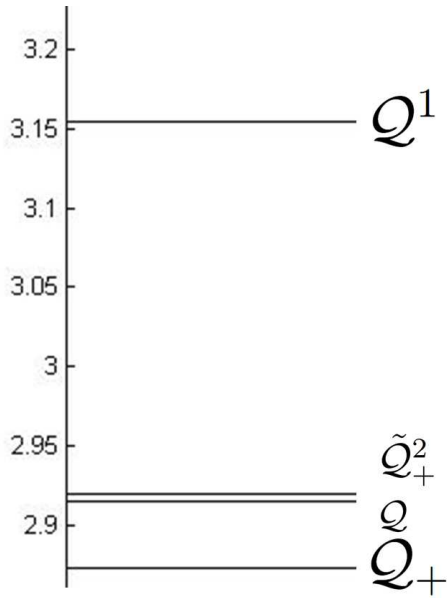}
  \caption{\textbf{Undistillable points of the CGLMP Bell inequality \cite{CGLMP}.} The bound for $\Q_{+}$ was obtained after optimizing over $\Q_{+}^2$ and finding the same result as conjectured in \cite{toni}. The set $\tilde{\Q}^2_{+}$ corresponds to an optimization where Condition \ref{positive3} in Definition \ref{Qplus_hier} was omitted.}
  \label{CGLMP_fig}
\end{figure}

\begin{figure}
  \centering
  \includegraphics[width=9.5 cm]{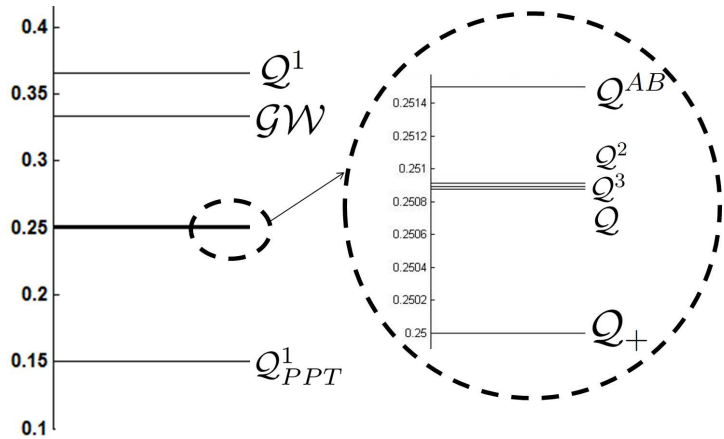}
  \caption{\textbf{Undistillable points of the $I_{3322}$ Bell inequality.}}
  \label{I3322_fig}
\end{figure}

\section{Conclusion}
\label{conclusion}

In this paper we have presented a zoo of physically consistent sets which admit an efficient characterization via linear and semidefinite programming. Our zoo defies the view that consistent sets are scarce or that their characterization, except for a few simple cases, is undecidable. We used different animals in the zoo to answer a number of questions regarding closed sets. Now we know that there exist closed non-trivial polytopes, infinitely many bipartite and multipartite supraquantum sets and closed sets with non-trivial intersection with the quantum set. As we saw, this last result implies the existence of bipartite physical principles $Z$ which are not stable under composition, in the sense that boxes belonging to different sets satisfying them can be wired into a box violating $Z$. This teaches us to beware of expressions like ``the maximal set of distributions compatible with Information Causality'', since, despite their seeming innocuity, they may not have any meaning at all. Important open questions which nevertheless remain open are the decidability of the nonlocality distillation problem (Question \ref{decidable}) and the existence/non-existence of a continuum of closed sets (Question \ref{continuum}).

To end with a positive note, the zoo has allowed us to derive, for the first time, a number of \emph{concrete} undistillable points of the
$I_{3322}$ \cite{I3322} and CGLMP \cite{CGLMP} inequalities. In view of this success, it would be interesting to find a continuous generalization of some of the families in the zoo which allow us to connect all such points.

\vspace{10pt}
\noindent\emph{Acknowledgements}

M.N. acknowledges the European Commission (EC) STREP "RAQUEL", as well as the MINECO project FIS2008-01236, with the support of FEDER funds. T.V. acknowledges financial support from a J\'anos Bolyai Grant of the Hungarian Academy of Sciences, the Hungarian National Research Fund OTKA (PD101461), and the T\'AMOP-4.2.2.C-11/1/KONV-2012-0001 project.


\begin{thebibliography}{99}
\bibitem{hardy}
L. Hardy, arXiv:quant-ph/0101012.
\bibitem{axioms}
B. Dakic and C. Brukner, \emph{Quantum Theory and Beyond: Is Entanglement Special?}, in \emph{Deep Beauty: Understanding the Quantum World through Mathematical Innovation}, 365-392, Ed. H. Halvorson, Cambridge University Press (2011).
\bibitem{lluis_mas}
Ll. Masanes and M. P. Mueller, New J. Phys. {\bf13}, 063001, (2011).
\bibitem{chiribella}
G. Chiribella, G. M. D'Ariano and P. Perinotti, Phys. Rev. A {\bf84}, 012311 (2011).
\bibitem{hardy2}
L. Hardy, arXiv:1104.2066.
\bibitem{PR_box}
S. Popescu and D. Rohrlich, Foundations of Physics {\bf24} (3): 379–385 (1994).
\bibitem{brassard}
G. Brassard, H. Buhrman, N. Linden, A. A. Methot, A. Tapp and F. Unger, F., Phys. Rev. Lett. {\bf96}, 250401, (2006).
\bibitem{linden}
N. Linden, S. Popescu, A. J. Short, and A. Winter, Phys. Rev. Lett. {\bf 99}, 180502 (2007).
\bibitem{info_caus}
M. Pawlowski, T. Paterek, D. Kaszlikowski, V. Scarani, A. Winter, and M. Zukowski, Nature {\bf461}, 1101 (2009).
\bibitem{mac_loc}
M. Navascu\'es and H. Wunderlich, Proc. Royal Soc. A 466:881-890 (2009).
\bibitem{loc_orth}
T. Fritz, A. B. Sainz, R. Augusiak, J. B. Brask, R. Chaves, A. Leverrier and A. Ac\'in, Nat. Commun. {\bf 4}, 2263 (2013),
e-print arXiv:1210.3018.
\bibitem{closed}
J. Allcock, N. Brunner, N. Linden, S. Popescu, P. Skrzypczyk and T. Vertesi, Phys. Rev. A {\bf80}, 062107 (2009).
\bibitem{short}
A. J. Short, Phys. Rev. Lett. {\bf102}, 180502 (2009).
\bibitem{forster}
M. Forster, Phys. Rev. A {\bf83}, 062114 (2011).
\bibitem{dukaric}
D. D. Dukaric and S. Wolf,  arXiv:0808.3317.



\bibitem{linear_prog}
E. D. Nering and A. W. Tucker, \emph{Linear Programs and Related Problems}, Academic Press (1993).
\bibitem{decidable}
T. Fritz, T. Netzer and A. Thom, arXiv:1207.0975.
\bibitem{quantum1}
M. Navascu\'es, S. Pironio, and A. Ac\'in, Phys. Rev. Lett. {\bf98}, 010401 (2007).
\bibitem{quantum2}
M. Navascu\'{e}s, S. Pironio and A. Ac\'{i}n, New J. Phys. {\bf 10}, 073013 (2008).
\bibitem{sdp}
L. Vandenberghe and S. Boyd, SIAM Review {\bf 38}, 49 (1996).

\bibitem{I3322_PV}
K. F. P\'al and T. V\'ertesi, Phys. Rev. A {\bf82}, 022116 (2010).
\bibitem{chsh}
J. F. Clauser, M. A. Horne, A. Shimony, and R. A. Holt, Phys. Rev. Lett. {\bf 23}, 880 (1969).
\bibitem{tsi_bound}
B. S. Cirel'son, Lett. Math. Phys. {\bf 4}, 93 (1980).
\bibitem{non_loc_dist}
N. Brunner and P. Skrzypczyk, Phys. Rev. Lett. {\bf102}, 160403 (2009).
\bibitem{eberhard}
P. H. Eberhard, Phys. Rev. A {\bf47}, 747 (1993).
\bibitem{liang}
Y.-C. Liang, T. V\'ertesi and N. Brunner, Phys. Rev. A, {\bf83}, 022108 (2011).
\bibitem{vidick}
T. Vidick and S. Wehner, Phys. Rev. A {\bf83}, 052310 (2011).
\bibitem{I3322}
M. Froissard, Nuov. Cim. B {\bf64}, 241 (1981).
\bibitem{junge}
M. Junge and C. Palazuelos, Comm. Math. Phys. {\bf306} (3), 695-746 (2011).
\bibitem{bin}
B. Yan, Phys. Rev. Lett. {\bf110}, 260406 (2013).
\bibitem{barbara}
B. Amaral, M. Terra Cunha and A. Cabello, arXiv:1306.6289.
\bibitem{tobias}
T. Fritz, A. Leverrier and A. B. Sainz,  arXiv:1212.4084.
\bibitem{cabello}
A. Cabello, Phys. Rev. Lett. {\bf110}, 060402 (2013).
\bibitem{extreme}
S. Pironio, J.-D. Bancal and V. Scarani, J. Phys. A: Math. Theor. {\bf44}, 065303 (2011).
\bibitem{TOBL}
R. Gallego, L. E. W\"{u}rflinger, A. Ac\'in and M. Navascu\'es, Phys. Rev. Lett. {\bf 107}, 210403 (2011).
\bibitem{gyni}
M. L. Almeida, J.-D. Bancal, N. Brunner, A. Ac\'in, N. Gisin and S. Pironio, Phys. Rev. Lett. {\bf104}, 230404 (2010).
\bibitem{ppt3}
T. V\'ertesi and N. Brunner, Phys. Rev. Lett. {\bf108}, 030403 (2012).
\bibitem{sliwa}
C. Sliwa, Phys. Lett. A {\bf 317}, 165 (2003).
\bibitem{lluis}
Ll. Masanes, A. Ac\'in and N. Gisin, Phys. Rev. A. {\bf73}, 012112 (2006).
\bibitem{pitowsky}
I. Pitowsky, Math. Program. A {\bf50}, 395-414 (1991).
\bibitem{QBell_NP}
T. Ito, H. Kobayashi and K. Matsumoto, 	arXiv:0810.0693v1.
\bibitem{moroder}
T. Moroder, J.-D. Bancal, Y.-C. Liang, M. Hofmann, O. G\"{u}hne, Phys. Rev. Lett. {\bf111}, 030501 (2013).
\bibitem{PPT}
A. Peres, Phys. Rev. Lett. {\bf 77}, 1413 (1996).
\bibitem{almost}
M. Navascu\'es, Y. Guryanova, M. Hoban and A. Ac\'in, article in preparation.
\bibitem{graph}
T. Fritz, A. Leverrier and A. B. Sainz,  arXiv:1212.4084.
\bibitem{klep}
S. Burgdorf and I. Klep, J. Operator Theory., {\bf68}, pp. 141-163 (2012).
\bibitem{siam}
S. Pironio, M. Navascu\'es and A. Ac\'in, SIAM J. Optim. {\bf20}, 5, 2157-2180 (2010).
\bibitem{CGLMP}
D. Collins, N. Gisin, N. Linden, S. Massar and S. Popescu, Phys. Rev. Lett. {\bf88}, 040404 (2002).
\bibitem{yalmip}
J. L\"{o}fberg, \emph{YALMIP : A Toolbox for Modeling and Optimization in MATLAB}. In `Proceedings of the CACSD Conference'', Taipei, Taiwan, 2004. http://users.isy.liu.se/johanl/yalmip/
\bibitem{sedumi}
J.F. Sturm, \emph{Using SeDuMi 1.02, a MATLAB toolbox for optimization over symmetric cones}, Optimization Methods and Software 11-12, 625-653 (1999). Special issue on Interior Point Methods. http://sedumi.mcmaster.ca.
\bibitem{scholz}
V. B. Scholz, R. F. Werner, arXiv:0812.4305.
\bibitem{connes1}
T. Fritz, Rev. Math. Phys. {\bf24(5)}, 1250012 (2012).
\bibitem{connes2}
M. Junge, M. Navascu\'es, C. Palazuelos, D. P\'erez-Garc\'ia, V. B. Scholz and R. F. Werner, J. Math. Phys. {\bf52}, 012102 (2011).
\bibitem{bancal}
J.D. Bancal, N. Gisin, S. Pironio, J. Phys. A: Math. Theor. {\bf 43}, 385303 (2010).
\bibitem{toni}
A. Ac\'in, T. Durt, N. Gisin and J. I. Latorre, Phys. Rev. A {\bf 65}, 052325 (2002).







\end{thebibliography}
\end{document}